\documentclass[a4paper,11pt]{article}

\usepackage[utf8]{inputenc}
\usepackage[T1]{fontenc}
\usepackage[english]{babel}

\usepackage{microtype}

\usepackage[pdftex]{graphicx}
\usepackage[usenames,dvipsnames]{xcolor}

\usepackage{amsmath,amssymb,amsfonts,amsthm}
\usepackage{mathabx}

\usepackage{fullpage}

\newcommand{\eps}{\varepsilon}
\newcommand{\Ot}{\tilde{O}}
\def\polylog{\operatorname{polylog}}
\def\poly{\operatorname{poly}}

\newtheorem{lemma}{Lemma}
\newtheorem{theorem}{Theorem}
\newtheorem{corollary}{Corollary}
\newtheorem{conjecture}{Conjecture}
\theoremstyle{definition}
\newtheorem{definition}{Definition}

\usepackage{authblk}
\title{On the Hardness of Partially Dynamic Graph Problems and Connections to
Diameter}

\author{Søren Dahlgaard\thanks{Research partly supported by Mikkel Thorup's
        Advanced Grant DFF-0602-02499B from the Danish Council for Independent
Research under the Sapere Aude research career programme.}}
\affil{University of Copenhagen\\\texttt{soerend@di.ku.dk}}
\date{}

\begin{document}

\maketitle

\begin{abstract}
    Conditional lower bounds for dynamic graph problems has received a great
    deal of attention in recent years. While many results are now known for the
    fully-dynamic case and such bounds often imply worst-case bounds for the
    partially dynamic setting, it seems much more difficult to prove amortized
    bounds for incremental and decremental algorithms. In this paper we
    consider partially dynamic versions of three classic problems in graph
    theory. Based on popular conjectures we show that:
    \begin{itemize}
        \item No algorithm with amortized update time $O(n^{1-\eps})$ exists
            for incremental or decremental maximum cardinality bipartite
            matching. This significantly improves on the $O(m^{1/2-\eps})$
            bound for sparse graphs of Henzinger et al. [STOC'15] and
            $O(n^{1/3-\eps})$ bound of Kopelowitz, Pettie and
            Porat\footnote{Kopelowitz et al. showed this
                result at SODA'16, and after posting their result online it was
                improved in an online version of the paper by Henzinger et al.
                Kopelowitz et al. also showed a slightly stronger bound of
            $O(n^{0.39-\eps})$ if node insertions are allowed.}. Our linear
            bound also appears more natural.
            In addition, the result we present separates
            the node-addition model from the edge insertion model, as an
            algorithm with total update time $O(m\sqrt{n})$ exists for the
            former by Bosek et al. [FOCS'14].
        \item No algorithm with amortized update time $O(m^{1-\eps})$ exists
            for incremental or decremental maximum flow in directed and
            weighted sparse graphs. No such lower bound was known for
            partially dynamic maximum flow previously. Furthermore no algorithm
            with amortized update time $O(n^{1-\eps})$ exists for directed and
            unweighted graphs or undirected and weighted graphs.
        \item No algorithm with amortized update time $O(n^{1/2 - \eps})$
            exists for incremental or decremental $(4/3-\eps')$-approximating
            the diameter of an unweighted graph. We also show a slightly stronger
            bound if node additions are allowed. The result is then extended to
            the static case, where we show that no $O((n\sqrt{m})^{1-\eps})$
            algorithm exists. We also extend the result to the case when an
            additive error is allowed in the approximation. While our bounds
            are weaker than the already known bounds of Roditty and Vassilevska
            Williams [STOC'13], it is based on a weaker conjecture of Abboud et
            al. [STOC'15] and is the first known reduction from the 3SUM
            and APSP problems to diameter.
            Showing an equivalence between APSP and diameter is a major open
            problem in this area (Abboud et al. [SODA'15]), and thus showing
            even a weak connection in this direction is of interest.
    \end{itemize}
\end{abstract}

\section{Introduction}
Arguably one of the most important goals of computer science is to understand
the complexity of natural computational problems. For many such problems we
know of polynomial time algorithms, but getting matching unconditional lower
bounds seem far beyond the scope of our current techniques. Therefore a recent
and very active line of research at top-level conferences concerns itself with
hardness results in the class \textbf{P}
\cite{Patrascu10,AbboudV14,AbboudVY15,HenzingerKNS15,KopelowitzPP16,AbboudGV15,ChechikLRSTW14,AbboudHVW16,CairoGR16,BringmannK15,BackursI15,AbboudBW15,AbboudBW15a}. Such
results are obtained by reducing from classic problems like 3SUM, APSP and
CNF-SAT, for which there exist very popular conjectures about the running time.
We call such a hardness result a \emph{conditional lower bound (CLB)} as it is
based (conditioned) on the truthfulness of some popular conjecture. The main
goal of CLBs is to explain barriers in algorithm development and provide
``warning signs'' that improving an algorithm for some problem has major and
surprising consequences for a classic problem like the ones mentioned above,
which researchers have worked on for decades, and trying to do so may be
ill-advised.

One particular area that has received a lot of attention from this perspective
is dynamic graph problems
\cite{RodittyZ11,Patrascu10,AbboudV14,AbboudVY15,HenzingerKNS15,KopelowitzPP16}.
In dynamic graph problems we are asked to maintain some property about a
graph such as reachability or shortest paths distances as the graph undergoes
changes (typically edge insertions and deletions). One may also consider the
partially dynamic cases where only edge insertions are allowed (incremental) or
edge deletions (decremental) or cases where node insertion and deletion is
allowed. Several conditional lower bounds are known for both partially and
fully dynamic problems such as shortest paths \cite{RodittyZ11,HenzingerKNS15},
maximum bipartite matching \cite{AbboudV14,HenzingerKNS15,KopelowitzPP16},
maximum flow \cite{AbboudVY15}, reachability
\cite{Patrascu10,AbboudV14,HenzingerKNS15}, and many more.

\subsection{Difficulties of partially dynamic}
Most of the research on CLBs for dynamic graph problems has been focused on the
fully dynamic case, however such results do not translate well into CLBs for
incremental or decremental algorithms. A typical reduction works by 1) building
a structured base graph, 2) for each element in some subset of the input
perform a series of insertions and queries to decide whether this element is in
a possible solution, 3) perform a series of deletions returning the graph to
its base state. From a partially dynamic perspective we may use the above
procedure to get similar \emph{worst-case} bounds, by keeping track of the data
structure state and simulating step 3 by \emph{rolling back} the insertions,
however this kills any hope of good amortized bounds. As noted in
\cite{AbboudV14,HenzingerKNS15,KopelowitzPP16} it seems more difficult to
obtain good bounds in this case, and specialized reductions are often
needed.

\subsection{Bounds under weaker assumptions}
While proving higher lower bounds is the main goal of CLBs, a simultaneous goal
is to prove similar CLBs under weaker assumptions, thus lending more credibility
to the belief that a problem is difficult or even impossible. Several recent
papers concerns themselves with this be either replacing a conjecture with a
weaker version as done by Abboud et al. in~\cite{AbboudHVW16} or by showing
similar reductions under several conjectures
\cite{VassilevskaW09,AbboudL13,AbboudVW14,AbboudVY15,HenzingerKNS15}.
As an example Abboud, Vassilevska Williams, and Yu~\cite{AbboudVY15} showed that
3SUM, APSP and CNF-SAT can all be reduced to the same problem of finding
triangles in a node-colored graph and showed several interesting results based
on the following conjecture:
\begin{conjecture}[\cite{AbboudVY15}]\label{conj:tc}
    At least one of the following is true:
    \begin{enumerate}
        \item There is no algorithm for the 3-SUM problem running in
            $O(n^{2-\eps})$ for any $\eps > 0$.
        \item There is no algorithm for the APSP problem on weighted graphs
            running in $O(n^{3-\eps})$ for any $\eps>0$.
        \item For every $\delta>0$ there is an integer $k\ge 3$ such that
            $k$-SAT on $n$ variables and $O(n)$ clauses cannot be solved in
            $2^{(1-\delta)n}\poly(n)$ time.
    \end{enumerate}
\end{conjecture}
The third item in Conjecture~\ref{conj:tc} is what is known as the strong
exponential time hypothesis (SETH) \cite{ImpagliazzoP01} and the $O(n)$ bound
on the number of clauses follows from the sparsification lemma of Impagliazzo,
Paturi, and Zane \cite{ImpagliazzoPZ01}.

\subsection{Our results}
In this paper we consider three of the perhaps most classic problems in graph
theory, namely maximum flow, maximum bipartite matching and diameter in the
partially dynamic setting.
For maximum flow and maximum bipartite matching we
show new, stronger, and more natural conditional lower bounds. For diameter we
show a new reduction from Conjecture~\ref{conj:tc} to both the partially
dynamic version of diameter and, perhaps more interestingly, the static case.
This is the first known connection from APSP and 3SUM to diameter in graphs and
addresses one of the main open problems in the area as stated in
\cite{AbboudGV15}.

\subparagraph*{Maximum bipartite matching}
In dynamic maximum cardinality bipartite matching we wish to maintain the size
of a maximum matching in a dynamic graph $G$. One can trivially do this in
$O(m)$ time by finding an augmenting path. Sankowski \cite{Sankowski07} gave a fully
dynamic algorithm with update time $O(n^{1.495})$ by using fast matrix
multiplication. In the incremental setting, one may consider a node-addition version in which the right-hand side of the
bipartite graph is given and the left-hand side arrives one node at a time with
all its incident edges. In this model Bosek et al.~\cite{BosekLSZ14} gave an
algorithm with total running time of $O(m\sqrt{n})$. From a hardness
perspective, Abboud and Vassilevska Williams \cite{AbboudV14} gave reductions
from 3SUM, triangle detection and boolean matrix multiplication to
fully-dynamic maximum cardinality bipartite matching. In particular, they
showed that a $O(n^{2-\eps})$ algorithm would imply a faster \emph{combinatorial} boolean
matrix multiplication algorithm. Their reductions, however, only imply
worst-case bounds in the case of partially dynamic algorithms. This was
addressed by Kopelowitz, Pettie and Porat~\cite{KopelowitzPP16} who revisited
P\v{a}tra\c{s}cu's reductions from \cite{Patrascu10} and showed that any
$O(n^{1/3-\eps})$ algorithm for incremental MCM would imply a truly
subquadratic algorithm for 3SUM. They also showed the same result for
$O(n^{0.39-\eps})$ algorithms when node insertions are allowed. Subsequently,
in an online version of \cite{HenzingerKNS15}, it was shown how to obtain a CLB
of $O(m^{1/2-\eps})$ in sparse graphs by reducing from the online matrix-vector
multiplication (OMv) problem.

In this paper we show the following theorem:
\begin{theorem}\label{thm:matching}
    There is no algorithm for solving incremental (or decremental) maximum
    cardinality
    bipartite matching with amortized time $O(n^{1-\eps})$ per insertion (or
    deletion) and $O(n^{2-\eps})$ time per query unless the OMv conjecture of
    \cite{HenzingerKNS15} is false.
\end{theorem}
One thing to note about Theorem~\ref{thm:matching} is that it separates the
node-addition model from the edge-insertion model as it implies a total running
time of $O(mn^{1-o(1)})$ in contrast to the $O(m\sqrt{n})$ running time of the algorithm from \cite{BosekLSZ14}. Furthermore, the reduction used to
prove Theorem~\ref{thm:matching} also rules out any efficient incremental (or
decremental) approximation algorithm that works by ruling out the existence of
short augmenting paths. Ruling out such paths is a popular way of ensuring a
good approximation ratio \cite{NeimanS13}.

\subparagraph*{Maximum flow}
Single-source single-sink maximum flow ($st$ Max-Flow) is one of the most
classic problems in graph theory.
In recent
years there have been several breakthrough results for $st$ maximum flow using
the powerful tools of Laplacian system solvers and interior point methods
\cite{Madry13,Sherman13,KelnerLOS14,LeeS14}. These algorithms seem to take
near-line time in practice, and the limits of our current analysis might be the
bottleneck in proving such upper bounds. Proving super-linear conditional
lower bounds for this problem may
thus be difficult if not impossible. Therefore, Abboud et al.~\cite{AbboudVY15}
considered different variants of the problem such as single-source maximum flow
and $ST$ maximum flow. They also showed that any algorithm solving the
fully-dynamic version of $st$ maximum flow with amortized update and query time
$O(n^{1-\eps})$ for any $\eps>0$ would refute Conjecture~\ref{conj:tc}.
Finally, we note that it is possible to modify the $m^{1-o(1)}$ CLB for fully
dynamic \#SSR of Abboud and Vassilevska Williams \cite{AbboudV14} to obtain a
$m^{1-o(1)}$ CLB for fully-dynamic $st$ max-flow in sparse graphs.

In this paper we show that even in the incremental and decremental case $st$
maximum flow exhibit the same kind of CLB, but based solely on SETH. This is
summarized in the following theorem:
\begin{theorem}\label{thm:maxflow2}
    There is no algorithm for solving incremental (or decremental) max $st$
    flow on a weighted and directed graph with $n$ nodes and $\Ot(n)$ edges with
    amortized time $O(m^{1-\eps})$ per operation for any $\eps > 0$ unless SETH
    is false.
\end{theorem}
Our bound shows that we cannot hope to get incremental maximum flow in offline
time as is the case for other problems.
We note that the above result only holds for directed
and weighted graphs. We show similar results for other types of graphs:
\begin{theorem}
    There is no algorithm for solving incremental (or decremental) max $st$
    flow on unweighted directed graphs or weighted undirected graphs on $n$
    nodes with amortized time $O(n^{1-\eps})$ per operation for any $\eps > 0$
    unless the OMv conjecture is false.
\end{theorem}
This result follows directly from Theorem~\ref{thm:matching} by using textbook
reductions from maximum bipartite matching to directed flow (see
e.g.~\cite{Cormen09book}) and from directed flow to undirected flow (see
e.g.~\cite{MadryThesis}).

\subparagraph*{Diameter}
The diameter problem asks us to compute the longest shortest-path distance in a
graph $G$. Efficiently computing or approximating the diameter is a basic
problem in graphs
\cite{AbboudGV15,AingworthCIM99,ChechikLRSTW14,CyganGS15,RodittyW13}. One can
trivially compute the diameter in the same time as computing APSP, however in
general no
better algorithm is known. It remains a major open problem whether a reduction
exists in the other direction \cite{AbboudGV15} - that is, can we compute all
distances in the same time as the longest? One can, however, approximate the
diameter faster. Roditty and Vassilevska Williams \cite{RodittyW13} showed how
to compute a $3/2$-approximation in time $\Ot(m\sqrt{n})$ randomized, and
Chechik et al.~\cite{ChechikLRSTW14} showed how to obtain the same guarantee
deterministically in time $\Ot(\min(m^{3/2},mn^{2/3}))$. More recently, it was
shown by Cairo, Grossi and Rizzi~\cite{CairoGR16} how to obtain a $(2 -
\frac{1}{2^k})$-approximation in time $\Ot(mn^{\frac{1}{k+1}})$. From a
hardness perspective it is known that any algorithm able to distinguish between
diameter $3$ and $2$ in time $O(m^{2-\eps})$ for sparse graphs would refute
SETH \cite{RodittyW13}. Chechik et al.~\cite{ChechikLRSTW14} showed that
approximating within a $4/3-\eps$ factor with additive error $\beta =
O(m^\delta)$ in time $O(m^{2-2\delta-\eps'})$ for sparse graphs would also
refute SETH, and this bound was improved in \cite{CairoGR16} to rule out any
$3/2-\eps$ approximation with the same additive error and time bounds based on
SETH (also for sparse graphs). From the perspective of dynamic algorithms
Abboud and Vassilevska Williams \cite{AbboudV14} showed that any algorithm for
$4/3-\eps$-approximating the diameter in a fully dynamic graph with amortized
update time $O(m^{2-\eps'})$ would refute SETH. We also note, that the above
static reductions rules out any $O(m^{1-\eps'})$ amortized update time for
incremental algorithms.

We note that all the reductions mentioned above are based on SETH. Similar to
the work of \cite{AbboudVY15,AbboudHVW16} we seek to replace this assumption by
a weaker one. In this paper we show the first reduction from 3SUM and APSP to
the diameter problem. That is, we show that a fast algorithm for approximating
the diameter implies a faster algorithm for the APSP and 3SUM
problems. The bounds we achieve are not as strong as the known bounds based on
SETH \cite{RodittyW13,ChechikLRSTW14,CairoGR16}, however they are based on a
weaker conjecture and hold even if SETH turns out to be false, thus giving more
credibility to the difficulty of the problem. For the partially dynamic case we
show the following theorem:
\begin{theorem}\label{thm:dyn_diam}
    There exists no incremental (or decremental) algorithm that approximates
    the diameter of an unweighted graph within a factor of $4/3-\eps$ running
    in amortized time $O(n^{1/2-\eps'})$ for any $\eps,\eps' > 0$ unless
    Conjecture~\ref{conj:tc} is false. Furthermore, if we allow node insertions
    in the incremental case the bound is $O(n^{0.618-\eps'})$.
\end{theorem}
In order to achieve the result for node insertions, we use the technique of
Kopelowitz et al.~\cite{KopelowitzPP16} leveraging rollback with our standard
incremental bound. By doing this we obtain a graph with fewer nodes
and thus a better bound.
More interestingly, we are able to generalize our results from the incremental
case to the following result for static graphs:
\begin{theorem}\label{thm:sta_diam}
    There exists no static $4/3-\eps$ approximation to the diameter on
    unweighted graphs running in $O((n\sqrt{m})^{1-\eps'})$ time for
    any $\eps,\eps' > 0$ and any number of edges $m$ unless
    Conjecture~\ref{conj:tc} is false.
\end{theorem}
As mentioned, this is the first known reduction from APSP to diameter and shows
at least some weak connection in this direction. An interesting property of
Theorem~\ref{thm:sta_diam} is that it holds \emph{for any} $m$ as a function of
$n$ and thus an algorithm need not exist \emph{for all} $m$.
As a corollary of Theorem~\ref{thm:sta_diam} we see that no algorithm can
$(4/3-\eps)$-approximate the diameter of static unweighted graph in time
$O(n^{2-\eps'})$ for any $\eps,\eps'$ unless Conjecture~\ref{conj:tc} is false.
This is reminiscent of the bounds from
\cite{RodittyW13,ChechikLRSTW14,CairoGR16}, however not quite as strong as it
does not hold for sparse graphs, for which we get a bound of $O(m^{3/2-\eps'})$.

Similar to \cite{ChechikLRSTW14,CairoGR16} we also extend the above bound to
the case of $(4/3-\eps)$-approximations with additive error $O(m^\delta)$. We
show the following
\begin{corollary}\label{cor:diam_add}
    There exists no static $4/3-\eps$ approximation with additive error
    $O(m^\delta)$ with running time $O(m^{\frac{3}{2}(1-\delta) - \eps'})$ or
    incremental/decremental algorithm with amortized time $O(m^{\frac{1}{2} -
    \frac{3\delta}{2} - \eps'})$ for any $\eps,\eps' > 0$ unless
    Conjecture~\ref{conj:tc} is false.
\end{corollary}

\subsection{A note on the decremental results and preprocessing}
We will in general only describe the reductions in the incremental case and
note that the decremental results are obtained by removing the edges in the
reverse order of insertions. This requires an assumption on the beginning
graph, and we will thus assume any suitable graph on $\Ot(n)$ edges in the
sparse case and the complete graph in the dense case.

Furthermore, we do not assume that any of the algorithms are allowed to
preprocess the graph. It is often an assumption in the design of amortized
partially dynamic algorithms that one starts with the empty (or complete) graph
in order for the analysis to work. Thus, our results hold for this case.

\section{Preliminaries}

\subparagraph*{Notation}
Throughout the paper we assume that matrices are boolean. Thus the output
of a vector-matrix-vector multiplication will always be a single bit. We
use $[n]$ to denote the set $\{0,\ldots,n-1\}$.

\subparagraph*{Online vector-matrix-vector multiplication}

We will consider the online vector-matrix-vector multiplication problem of
\cite{HenzingerKNS15}:

\begin{definition}[OuMv problem \cite{HenzingerKNS15}]\label{defn:oumv}
    Let $M$ be a binary $n\times n$ matrix than can be preprocessed. After
    preprocessing $n$ vector pairs $(u^1,v^1), \ldots, (u^n,v^n)$ arrive one at
    a time and the task is to compute $(u^i)^T M v^i$ before being presented with
    the $i+1$th vector pair for every $i$.
\end{definition}

In \cite{HenzingerKNS15} they showed that the OMv problem can be reduced to the
OuMv problem. They also came up with the following conjecture:

\begin{conjecture}[\cite{HenzingerKNS15}]\label{conj:omv}
    There is no algorithm for the OMv problem (and thus the OuMv problem)
    running in time $O(n^{3-\eps})$ for any $\eps>0$.
\end{conjecture}

\subparagraph*{Triangle collection}

We will also consider the triangle collection problem of \cite{AbboudVY15}:

\begin{definition}[Triangle collection \cite{AbboudVY15}]
    Given a node-colored graph $G$, is it true that for every triplet of colors
    $a,b,c$ there exists a triangle $(u,v,w)$ in $G$ where $u$ has color $a$,
    $v$ has color $b$ and $w$ has color $c$?
\end{definition}

In fact, we will consider the more structured triangle collection* (TC*)
problem which they also used in \cite{AbboudVY15}

\begin{definition}[Triangle collection* \cite{AbboudVY15}]
    Let $n,\Delta,p$ be parameters and let $G$ be an undirected node-colored
    tripartite graph with partitions $A,B,C$. Let $G$ be any graph with the
    following structure:
    \begin{itemize}
        \item Each partition has its own $n$ colors and we denote these by the
            numbers of $[n]$ for each partition.
        \item $A$ contains nodes of the form $a^i_j$, where $i\in [n]$ is the
            color of the node and $j\in [\Delta]$.
        \item $B$ and $C$ contains nodes of the form $b^i_{j,x}$ and
            $c^i_{j,x}$ where $i\in [n]$ is the color of the node and $j\in
            [\Delta], x\in [p]$.
    \end{itemize}
    And the edges of $G$ are as follows:
    \begin{itemize}
        \item For each $i,i'\in [n]$ and $j\in [\Delta]$ there is an edge from
            $a^i_j$ to $b^{i'}_{j,x}$ for \emph{exactly one} $x$. Similarly
            there is an edge from $a^i_j$ to $c^{i'}_{j,y}$ for exactly one $y$
            (note that $y$ and $x$ need not be the same for the same $j$ and
            $i'$).
        \item There may be an edge between nodes $b^i_{j,x}$ and
            $c^{i'}_{j',y}$ \emph{only if} $j=j'$.
    \end{itemize}
    We ask the following question: Does there exist a triple of colors (one
    color per partition) such that $G$ does not contain a triangle with these
    colors?
\end{definition}

In \cite{AbboudVY15} it was shown that this problem does not have a truly
subcubic algorithm unless Conjecture~\ref{conj:tc} is false.

It will be important that the reductions from these problems to TC* hold even
when $\Delta$ and $p$ are bounded by $\polylog(n)$.

\section{Incremental maximum matching}
We will reduce from the OuMv problem of Definition~\ref{defn:oumv}.
Observe that the OuMv problem is equivalent to the following statement:
For each vector pair $u^i,v^i$ determine whether indices $j,k$ exist, such that
$u^i_j = v^i_k = M_{jk} = 1$. In order to model this as an incremental maximum
matching problem we construct the following graph:
Create $6$ copies of $2n$ nodes and name these $S, A, B, C, D, T$.
Partition $A$ into $n$ pairs of nodes $a^\ell_1, a^r_1, \ldots, a^\ell_n,
a^r_n$. Do the same for $S, B, C, D, T$. Add the edges $(a^\ell_i, a^r_i)$ for
each $i$ and do the same for $B, C, D$. Now for each $i,j$ add the edge $(b^r_i,
c^\ell_j)$ if $M_{ij} = 1$. Observe that this graph has a unique maximum
matching each $(\ell,r)$ pair. Observe also that the graph is bipartite.
Now we do the following $n$ phases -- one for each $u^i,v^i$ vector pair.
\begin{enumerate}
    \item For each $j$ such that $u^i_j = 1$ add the edge $(a_i^r,
        b^\ell_j)$.
    \item For each $j$ such that $v^i_j = 1$ add the edge
        $(c_j^r,d_i^\ell)$.
    \item Add the edges $(s^r_i,a^\ell_i)$ and $(d^r_i,t^\ell_i)$.
    \item Query the size of a maximum matching.
    \item Add the edges $(s^\ell_i,s^r_i)$ and $(t^\ell_i,t^r_i)$.
\end{enumerate}
This is illustrated in Figure~\ref{fig:matching}.

\begin{figure}[htbp]
    \centering
    \includegraphics[width=.7\textwidth]{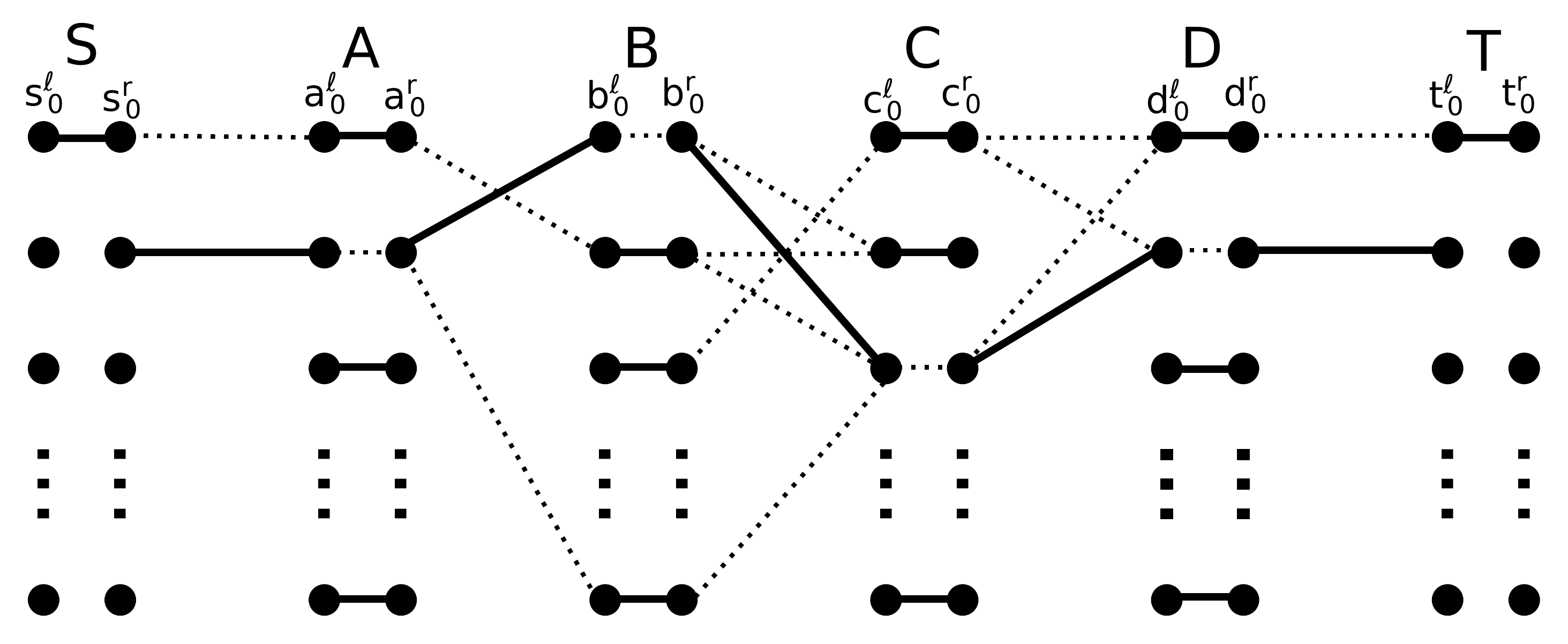}
    \caption{Reduction to incremental maximum matching.}
    \label{fig:matching}
\end{figure}

\begin{lemma}\label{lem:max_match}
    Let the setting be as above and let the phases be numbered $0, 1, \ldots,
    n-1$. Then the size of the maximum matching during the $i$th phase is exactly
    $4n + 2i + 1$ if the resulting vector-matrix-vector product is $1$ and
    $4n+2i$ otherwise.
\end{lemma}
\begin{proof}
    Note that prior to any of the $i$ phases the size of the maximum matching
    is exactly $4n + 2i$, which is also a perfect matching of the graph induced
    by the edges. To see this observe that each $s^\ell_0, \ldots,
    s^\ell_{i-1}$ must be matched to its corresponding
    $s^r_0,\ldots,s^r_{i-1}$, and this is the only edge incident to the
    $\ell$-nodes. As a consequence of this, each $a^\ell_j$ must be matched
    with $a^r_j$, and so on.

    Now consider the $i$th phase. Adding any edge $(a^r_i, b^\ell_j)$ or
    $(c^r_j,d^\ell_i)$ cannot increase the size of the maximum matching, as the
    size of the subgraph induced by the edges of the graph does not increase --
    i.e. all nodes with edges incident to them are already matched.

    Assume that adding the edges $(s^r_i,a^\ell_i)$ and $(t^\ell_i,d^r_i)$
    increases the matching. The matching can increase by at most $1$, as only
    two more nodes can be matched. Furthermore the matching must now contain
    edges as follows
    \[
        (s^r_i,a^\ell_i), (a^r_i,b^\ell_j), (b^r_j, c^\ell_k), (c^r_k,
        d^\ell_x), (d^r_x, t^\ell_y)\ .
    \]
    Now observe that each $t^\ell_y$ for $y < i$ \emph{must} be matched to
    $t^r_y$, as the right nodes have no other incident edges and all nodes have
    to be matched for the size of the matching to increase. Thus we must have
    $y=i$ in the list above, but this means that we have exactly found a pair
    $j,k$ such that $u^i_j = v^i_k = M_{jk} = 1$ and the vector-matrix-vector
    product is thus $1$.

    Conversely, assume that the vector-matrix-vector product is $1$, then such
    an index pair $j,k$ must exist and we can find the following matching of
    size $4n+2i+1$: Match the edges
    \[
        (s^r_i,a^\ell_i), (a^r_i,b^\ell_j), (b^r_j, c^\ell_k), (c^r_k,
        d^\ell_i), (d^r_i, t^\ell_i)\ .
    \]
    For all $x < i$ add the edges $(s^\ell_x,s^r_x)$ and $(t^\ell_x,t^r_x)$ to
    the matching. For all $x\ne i$ add the edges $(a^\ell_x,a^r_x)$ and
    $(d^\ell_x,d^r_x)$ to the matching. And for all $x\ne j$ and $y\ne k$ add
    the edges $(b^\ell_x,b^r_x)$ and $(c^\ell_y,c^r_y)$ to the matching. This
    matches all nodes incident to an edge and has size $4n+2i+1$. This is also
    exactly the matching illustrated in Figure~\ref{fig:matching} for $i=1$.
\end{proof}

It follows from Lemma~\ref{lem:max_match} that we can solve the OuMv problem
correctly via this reduction. The reduction creates a graph with $O(n)$ nodes
and $O(n^2)$ edges. We perform $O(n^2)$ insertions and $O(n)$ queries giving
the result in Theorem~\ref{thm:matching}

\section{Maximum flow}
In order to show Theorem~\ref{thm:maxflow2} we will use a similar graph
construction as have been used numerous times before
\cite{PatrascuW10,RodittyW13,ChechikLRSTW14,AbboudV14}:
First partition the variables of the SAT problem into two
groups $A$ and $B$ of $n/2$ variables each. For each possible assignment to the variables
in $A$ we create a node in our graph $G$ (and likewise for $B$). Furthermore,
for each clause of the SAT formula, we create a node as well. We denote the
corresponding sets of nodes by $A, B, C$. Set $N = 2^{n/2} = |A| = |B|$. For
each pair of nodes $a\in A, c\in C$ we add the directed edge $(a,c)$ with
capacity $N$ if the partial assignment $a$ \emph{does not} satisfy the clause
$c$. Similarly, for each pair of nodes $b\in B, c\in C$ we add the directed
edge $(c,b)$ with capacity $1$ if $b$ does not satisfy $c$. Finally we add two
nodes $s,t$ and add edges $(b,t)$ with capacity $1$ for each $b\in B$.

We now continue in phases with a phase for each $a\in A$. Denote these nodes
by $a_1,a_2,\ldots,a_N$:
\begin{enumerate}
    \item Add the edge $(s,a_i)$ with capacity $N$.
    \item Query the maximum flow between $s$ and $t$.
    \item Add the edge (``shortcut'') $(a_i,t)$ with capacity $N$.
\end{enumerate}
This is illustrated in Figure~\ref{fig:maxflow}.
\begin{figure}[htbp]
    \centering
    \includegraphics[width=.6\textwidth]{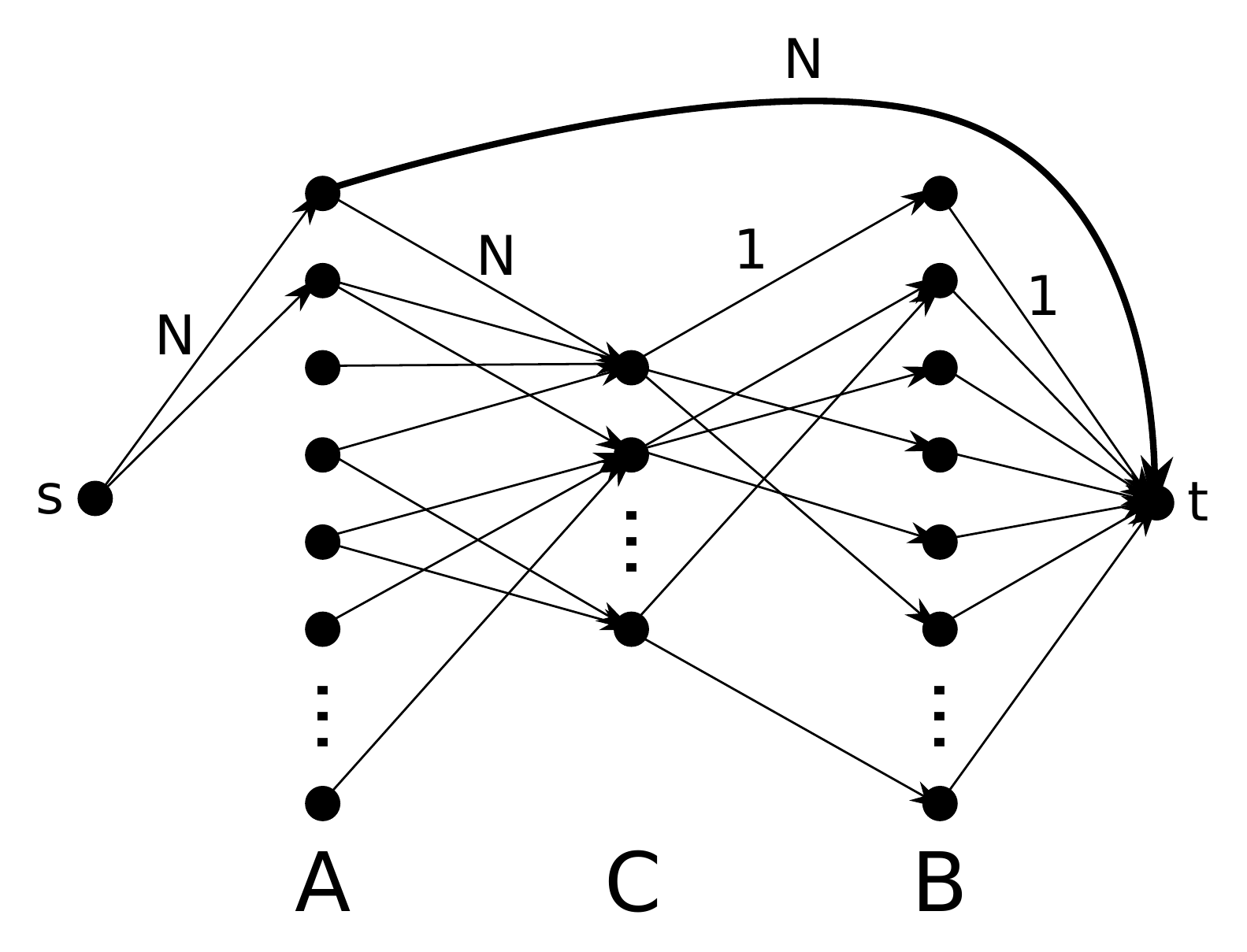}
    \caption{Illustration of the incremental construction for maximum flow.}
    \label{fig:maxflow}
\end{figure}

\begin{lemma}\label{lem:maxflow_correct}
    Let the setup be as described above. If the $st$ flow returned during any
    of the $i$ phases is $< i\cdot N$, then the SAT formula is satisfiable.
    Otherwise the formula is not satisfiable.
\end{lemma}
\begin{proof}
    Observe, that prior to the $i$th phase, the flow is exactly $(i-1)\cdot N$,
    as we can use the paths $(s,a_j), (a_j, t)$ for each $j< i$, which has
    capacity $N$ and exactly $(i-1)\cdot N$ flow leaves $s$.

    Now assume that the partial formula corresponding to $a_i$ can be completed
    to a satisfying assignment. In this case, there must be some node $b\in B$,
    for which there is no path from $a_i$ to $b$. This follows because such a
    path has to go through a node $c\in C$, but then both $a_i$ and $b$ do not
    satisfy the clause $c$, which is a contradiction. However, the only way to
    send flow from $a_i$ to $t$ is through the nodes $b\in B$ and thus it is
    not possible to send all $N$ units of flow from $a_i$ to $t$.

    Now assume that the flow is $<i\cdot N$, then there must be
    some $b\in B$ such that there is no path from $a_i$ to $b$. Otherwise, we
    could route $N$ units of flow from $a_i$ to $t$ via the nodes of $B$ and
    the remaining $(i-1)\cdot N$ units through the ``shortcuts''. It now
    follows that $a_i$ and $b$ together satisfy all clauses (otherwise there
    would be a path) and thus the CNF formula is satisfiable.

    Since this is true for all of the $i$ phases, the statement of the lemma
    follows.
\end{proof}

As a consequence of Lemma~\ref{lem:maxflow_correct} we may use the above
procedure to solve the given SAT problem. By the sparsification lemma of
\cite{ImpagliazzoPZ01} it follows that we can assume the graph has $O(N)$ nodes
and $\Ot(N)$ edges and we perform a total of $\Ot(N)$ insertions and queries.
The result of Theorem~\ref{thm:maxflow2} thus follows directly.

\section{Diameter}
In this section we show how to obtain conditional lower bounds for the problem
of approximating the diameter of an unweighted graph within a factor of
$4/3-\eps$.

\subsection{A graph construction}
We will first describe the graph structure we use.

\begin{definition}
    Let $G$ be an instance of the TC* problem as defined above. We will define
    the graph $H_{\gamma,k}(G)$. The idea is that $H_{\gamma,k}(G)$
    ``corresponds'' to the colors $\{kn^\gamma, \ldots,(k+1)n^\gamma -1\}$ of $A$.
    Thus $k$ is a number in $[n^{1-\gamma}]$. The nodes of this graph are as
    follows:
    \begin{itemize}
        \item The nodes $B$ and $C$ of $G$.
        \item For each color $i\in \{kn^\gamma, (k+1)n^\gamma -1\}$ of $A$ we
            add the nodes $a^i_0,\ldots, a^i_{n-1}$ and
            $t^i_0,\ldots,t^i_{n-1}$.
        \item We also add several special nodes: A ``master node'' $u$,
            $n^\gamma$ ``skip nodes'' $v_i$ and three ``connector nodes''
            $w_1,w_2,w_3$.
    \end{itemize}
    For a color $i\in \{kn^\gamma,(k+1)n^\gamma-1\}$ we denote the nodes
    $a^i_0,\ldots a^i_{n-1}$ by $A_i$ and the collection of all $A_i$s by $A$.
    We do the same for $T_i$ and $T$.

    The edges of $H_{\gamma,k}(G)$ are as follows:
    \begin{itemize}
        \item Add the edges between $B$ and $C$ in $G$.
        \item Connect the node $w_1$ to each node of $A$ and $w_2$.
        \item Connect $w_2$ to each node of $B$ and $C$ as well as $w_3$ and
            the master node $u$.
        \item Connect $w_3$ to each node of $T$.
        \item Connect $u$ to all nodes $v_i$.
        \item For each $i\in \{kn^\gamma, (k+1)n^\gamma -1\}$ do as follows:
            \begin{itemize}
                \item Connect $v_i$ to all nodes of $T\setminus T_i$ and to all
                    nodes of $A_i$.
                \item For each $i'\in [n]$ and each edge $(a^i_j,
                    b^{i'}_{j,x})\in G$ add the edge $(a^i_{i'},
                    b^{i'}_{j,x})$.
                \item For each $i'\in [n]$ and each edge $(a^i_j,
                    c^{i'}_{j,x})\in G$ add the edge $(c^{i'}_{j,x},
                    t^i_{i'})$.
            \end{itemize}
    \end{itemize}
\end{definition}
An overview of the graph $H_{\gamma,k}(G)$ is illustrated in Figure~\ref{fig:34diam}
and a more detailed view in Figure~\ref{fig:tc_diam_det}.

\begin{figure}[htbp]
    \centering
    \includegraphics[width=.7\textwidth]{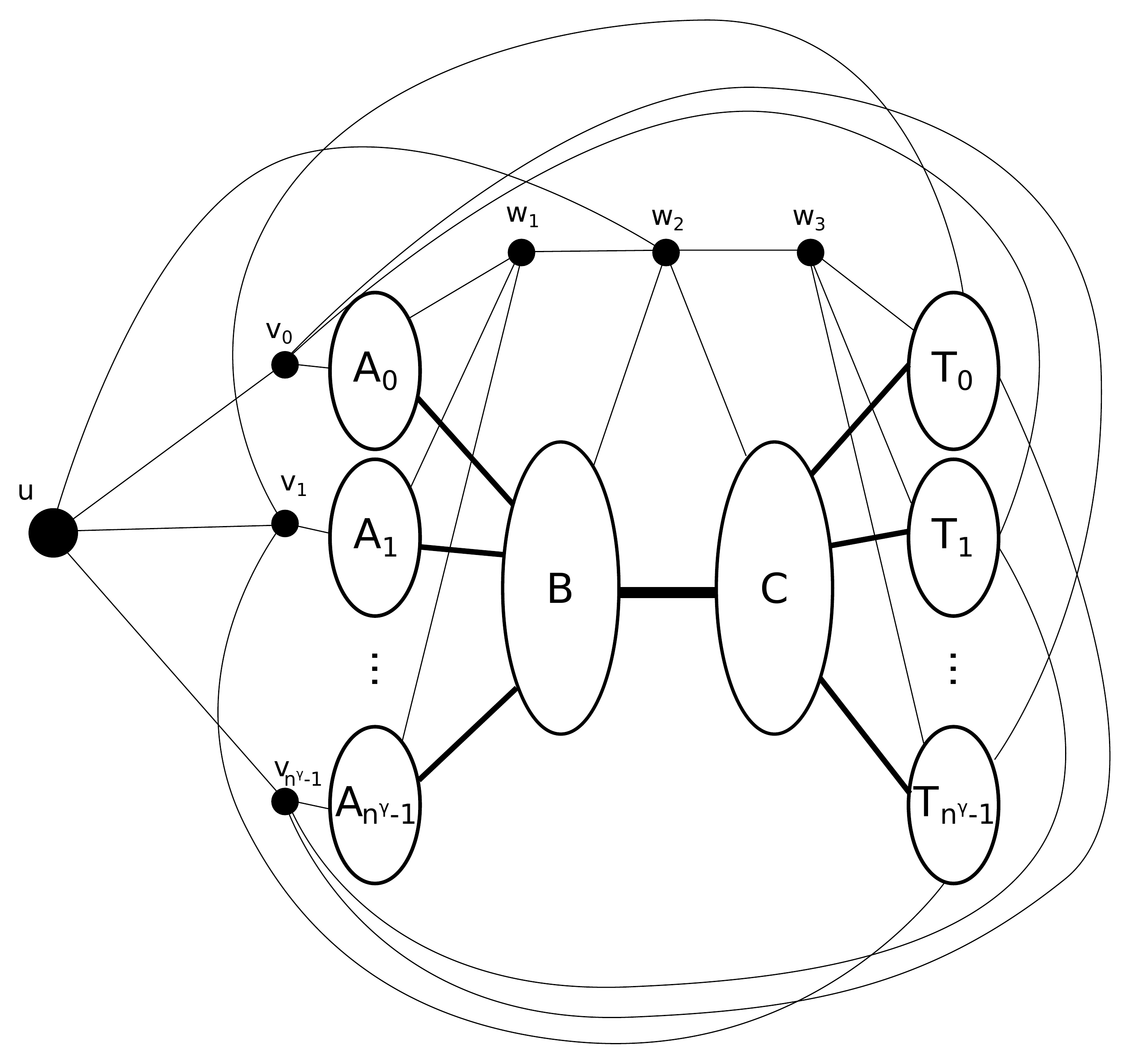}
    \caption{Diameter structure.}
    \label{fig:34diam}
\end{figure}
\begin{figure}[htbp]
    \centering
    \includegraphics[width=.6\textwidth]{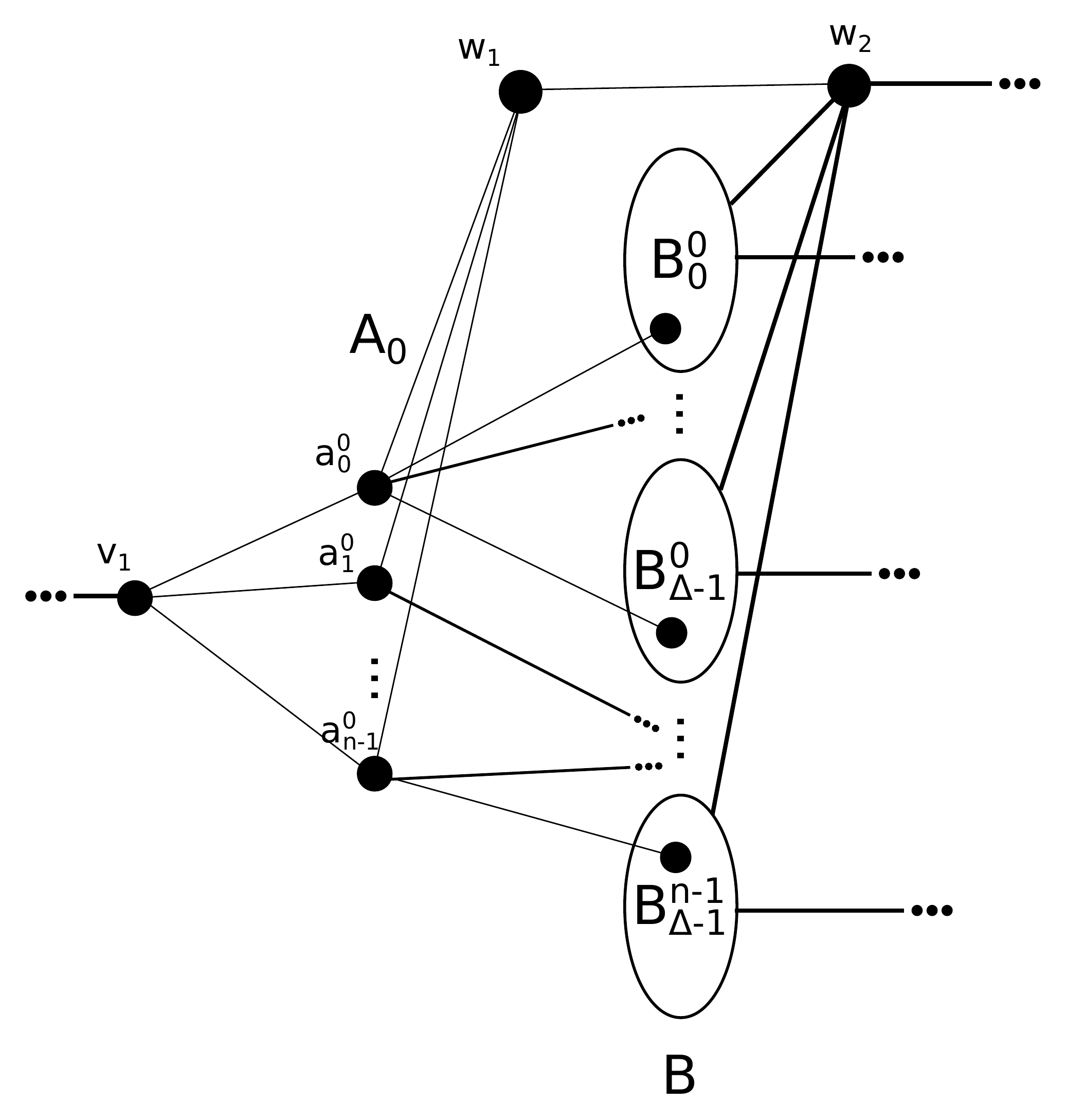}
    \caption{Diameter structure.}
    \label{fig:tc_diam_det}
\end{figure}

The idea is that length three paths between $A_i$ and $T_i$ correspond to
triangles in $G$ containing the color $i$ of $A$. Each of the $n$ nodes in
$A_i$ thus correspond to picking a color from $B$ and each of the $n$ nodes in
$T_i$ correspond to picking a color from $C$. If two such nodes don't have a
length three path there is no triangle in $G$ of the corresponding triplet of
colors. In this case the connector nodes ensure that there is a length four
path between the nodes. The master and skip nodes ensure that all other nodes
have distance at most $3$. This is captured by the following lemma:

\begin{lemma}\label{lem:diam_graph}
    Let $G$ be an instance to the TC* problem and let $H_{\gamma,k}(G)$ be as
    defined above. Let $i\in \{kn^\gamma, (k+1)n^\gamma-1\}$ be a color of $A$
    and let $\alpha,\beta\in [n]$ be colors of $B$ and $C$ respectively. Then
    the distance from $a^i_\alpha$ to $t^i_\beta$ in $H_{\gamma,k}(G)$ is $3$
    if the colors $i,\alpha,\beta$ have a triangle in $G$ and $4$ otherwise.
\end{lemma}
\begin{proof}
    Assume first that there is a triangle $a^i_j, b^\alpha_{j,x},
    c^\beta_{j,y}$ for some $j$ in $G$ (note that such a triangle can only
    occur if $j$ is the same for all the three nodes). In this case there is a
    path $a^i_\alpha, b^\alpha_{j,x}, c^\beta_{j,y}, t^i_\beta$ in
    $H_{\gamma,k}(G)$ and thus the distance is at most $3$. Observe also, that
    no node is connected to both $A_i$ and $T_i$ and thus the distance is
    strictly greater than $2$.

    Now assume that the distance from $a^i_\alpha$ to $t^i_\beta$ is $3$. Such
    a path has to go from $A_i$ to $B$ to $C$ to $T_i$ as any node $w_\ell,
    v_\ell$ or $u$ either has distance $3$ to one of $a^i_\alpha$ or
    $t^i_\beta$ or it has distance $2$ to both of them. Now consider a shortest
    path $a^i_\alpha, b, c, t^i_\beta$, where $b$ and $c$ are the nodes of $B$
    and $C$ on this path. Clearly the node $b$ must have color $\alpha$ in $G$
    as it would not have an edge to $a^i_\alpha$ otherwise, and similarly $c$
    must have color $\beta$ in $G$. Thus the path consists of nodes
    $a^i_\alpha, b^\alpha_{j,x}, c^\beta_{j',y}, t^i_\beta$. Since no edge in
    $G$ goes between nodes with different $j$-values we must have $j' = j$. It
    is now clear that the edge $(a^i_\alpha, b^\alpha_{j,x})$ corresponds to
    the edge $(a^i_j,b^\alpha_{j,x})$ in $G$ and the edge
    $(c^\beta_{j,y},t^i_\beta)$ corresponds to the edge $(a^i_j,c^\beta_{j,y})$
    in $G$. Thus, these three nodes form a triangle of the correct color triple
    in $G$.
\end{proof}
Furthermore it is easy to see that the longest distance in $H_{\gamma,k}(G)$ is
at most $4$, thus the diameter is $4$ exactly when one of the corresponding
color triplets do not have a triangle in $G$ and $3$ otherwise.

\subsection{Dynamic}
We will first consider the problem without node additions. For simplicity we
only consider the incremental case and note that the decremental case follows
by deleting edges until we obtain the same graph\footnote{Under the assumption
that the algorithm starts with some suitable graph}.

Given an instance to the TC* problem we create the graph $H_{1,0}(G)$ (that is,
the graph representing all colors of $A$). This graph is created by adding
edges incrementally and has $\Ot(n^2)$ nodes and edges. It follows that an edge
insertion must take $n^{1/2-o(1)}$ time unless Conjecture~\ref{conj:tc} is
false.

Next, we consider the problem with node additions. It was shown in
\cite{KopelowitzPP16} that if we allow node additions in the problem of
incremental maximum matching, it is possible to show stronger lower bounds by
leveraging the amortized running time with the widely used rollback technique.
We here apply the same argument to the problem of incremental diameter
approximation.

The goal is again to construct (a subgraph of) $H_{1,0}(G)$ but we do not start
with all nodes in the graph. We will assume that the amortized
running time of an insert operation is $n^\alpha$ for some $\alpha$. The goal
is to get a bound on $\alpha$ by expressing the total running time in terms of
$\alpha$ and using the assumption on running time for TC*. We let $\hat{n}$
denote the current number of nodes in the graph $G$. We continue as follows:
\begin{enumerate}
    \item Insert all nodes of $B$ and $C$ into the dynamic graph. Also insert
        the nodes $w_1, w_2, w_3$ and $u$. We also insert all the edges
        induced by these nodes in $H_{1,0}(G)$ into the graph.
    \item For each color $i\in [n]$ of $A$ we do a phase:
        \begin{itemize}
            \item We insert the nodes of $A_i, T_i, v_i$ into the dynamic graph
                and all the edges induced by these nodes and the current state
                of the dynamic graph in $H_{1,0}(G)$.
            \item Query the diameter of the graph.
            \item Assume we inserted $k$ edges+nodes in this phase. If the
                total running time of all these insertions was greater than
                $2k\hat{n}^\alpha$ we keep the nodes in the graph. Otherwise we
                rollback all operations of this phase.
        \end{itemize}
\end{enumerate}
We answer the question of the TC* problem according to whether the diameter was
$3$ all the time or not similar to the proof of the case without
node additions.

The goal is now to bound $\alpha$ by using the method of \cite{KopelowitzPP16}. We will do
this by carefully counting the number of ``amortized credit units'' the data
structure has and using this to bound the total number of nodes added to the
graph (i.e. not rolled back).

Observe that after the first step, we have added $\Ot(n^2)$ edges to the graph
and $\Ot(n)$ nodes. Thus the data structure has at most $\Ot(n^{2+\alpha})$
credit at this point (this happens if almost all operations were $O(1)$). Now
consider the total time spent by the algorithm. This can be bounded by
$\Ot(n^2\cdot N^\alpha)$ where $N$ is the number of nodes at the end of all
phases. This is the case since $N\ge N_0$, where $N_0 = \Ot(n)$ is the number
of nodes after the first step and there are at most $\Ot(n^2)$ total
operations. Note that this would not be the case if we did not have a bound on
the cost of the rolled back operations, but we only rollback the cheap
operations, so this is okay. We wish to express $N$ in terms of $n$ and
$\alpha$ in order to express the total running time in terms of these.

Observe, that every time we keep the added nodes in the graph, the data
structure spent at least twice the amortized cost. Since
we started out with $\Ot(n^{2+\alpha})$ credit it must be true that
\[
    \sum_{i=N_0}^N i^\alpha \le \text{cost of non-rollbacked operations} = \Ot(n^{2+\alpha})\ .
\]
The worst case is if $N$ is polynomially larger than $N_0$, and thus
$\sum_{i=N_0}^N i^\alpha = \Omega(N^{1+\alpha})$. It follows that $N =
\Ot(n^{\frac{2+\alpha}{1+\alpha}})$.
Thus the total running time is $\Ot(n^2\cdot n^{\frac{2+\alpha}{1+\alpha}\alpha})$. Now, by Conjecture~\ref{conj:tc} we must have $\frac{2+\alpha}{1+\alpha}\alpha
= 1 - o(1)$. Solving this for $\alpha$ gives $\alpha = \frac{\sqrt{5}-1}{2} <
0.618$.

\subsection{Static}
\begin{proof}[Proof of Theorem~\ref{thm:sta_diam}]
    Let $G$ be an instance of the TC* problem with parameters $n,\Delta,p$ with
    $\Delta$ and $p$ bounded by $\Ot(1)$ and $m = \Ot(n^2)$ as in
    \cite{AbboudVY15}.

    For a parameter $0 < \gamma\le 1$ we
    create the graphs $H_{\gamma, 0}(G), \ldots, H_{\gamma, n^{1-\gamma}
    -1}(G)$ and solve the diameter problem on these graphs up to a $4/3-\eps$
    approximation. This is sufficient to distinguish between diameters $4$ and
    $3$ in all of the graphs. Now, if the diameter is $4$ in just one of the
    graph we answer that there exists a triplet of colors such that there is no
    triangle in $G$. This follows from Lemma~\ref{lem:diam_graph}.

    We note that the graphs $H_{\gamma, k}(G)$ each have $N = \Ot(n^{1+\gamma})$
    nodes and $M = \Ot(n^2)$ edges. Assume now that that any algorithm
    approximating the diameter within a factor of $4/3-\eps$ in time
    $O(N\sqrt{M}^{1-\eps'}) = O(n^{2+\gamma-\eps'})$ for any $\eps,\eps' > 0$
    exists. Since we create $n^{1-\gamma}$ instances of the problem this would
    imply an $O(n^{3-\eps''})$ algorithm for the TC* problem for some $\eps'' >
    0$.
\end{proof}

\subsection{Additive error}
To see Corollary~\ref{cor:diam_add} we fix $m^\alpha$ and consider TC* on a
graph $G$ with $N$ nodes and $M = \Ot(N^2)$ edges such that $M = m^{1-\alpha}$.
We then create $H_{1, 0}(G)$ and subdivide each edge into $m^\alpha$ nodes.
This graph now has $m$ nodes and edges and any algorithm solving $4/3-\eps$
diameter with additive error $O(m^\alpha)$ in time $M^{3/2-\eps'} =
m^{\frac{3}{2}(1-\alpha) - \eps''}$ time thus violates
Conjecture~\ref{conj:tc}.

\subparagraph*{Acknowledgements}
I would like to thank Amir Abboud and Virginia Vassilevska Williams for helpful
discussions and observations. I would also like to thank an anonymous referee
for pointing out the reduction from incremental matching to undirected flow.


\bibliographystyle{plain}
\bibliography{cond_lbs}

\begin{thebibliography}{10}

\bibitem{AbboudBW15a}
Amir Abboud, Arturs Backurs, and Virginia~Vassilevska Williams.
\newblock If the current clique algorithms are optimal, so is valiant's parser.
\newblock In {\em Proc. 56th IEEE Symposium on Foundations of Computer Science
  (FOCS)}, pages 98--117, 2015.

\bibitem{AbboudBW15}
Amir Abboud, Arturs Backurs, and Virginia~Vassilevska Williams.
\newblock Tight hardness results for {LCS} and other sequence similarity
  measures.
\newblock In {\em Proc. 56th IEEE Symposium on Foundations of Computer Science
  (FOCS)}, pages 59--78, 2015.

\bibitem{AbboudGV15}
Amir Abboud, Fabrizio Grandoni, and Virginia~Vassilevska Williams.
\newblock Subcubic equivalences between graph centrality problems, {APSP} and
  diameter.
\newblock In {\em Proc. 26th ACM/SIAM Symposium on Discrete Algorithms (SODA)},
  pages 1681--1697, 2015.

\bibitem{AbboudHVW16}
Amir Abboud, Thomas~Dueholm Hansen, Virginia~Vassilevska Williams, and Ryan
  Williams.
\newblock Simulating branching programs with edit distance and friends or: A
  polylog shaved is a lower bound made.
\newblock In {\em Proc. 48th ACM Symposium on Theory of Computing (STOC)},
  2016.
\newblock To appear.

\bibitem{AbboudL13}
Amir Abboud and Kevin Lewi.
\newblock Exact weight subgraphs and the k-sum conjecture.
\newblock In {\em Proc. 40th International Colloquium on Automata, Languages
  and Programming (ICALP)}, pages 1--12, 2013.

\bibitem{AbboudV14}
Amir Abboud and Virginia~Vassilevska Williams.
\newblock Popular conjectures imply strong lower bounds for dynamic problems.
\newblock In {\em Proc. 55th IEEE Symposium on Foundations of Computer Science
  (FOCS)}, pages 434--443, 2014.

\bibitem{AbboudVW14}
Amir Abboud, Virginia~Vassilevska Williams, and Oren Weimann.
\newblock Consequences of faster alignment of sequences.
\newblock In {\em Proc. 41st International Colloquium on Automata, Languages
  and Programming (ICALP)}, pages 39--51, 2014.

\bibitem{AbboudVY15}
Amir Abboud, Virginia~Vassilevska Williams, and Huacheng Yu.
\newblock Matching triangles and basing hardness on an extremely popular
  conjecture.
\newblock In {\em Proc. 47th ACM Symposium on Theory of Computing (STOC)},
  pages 41--50, 2015.

\bibitem{AingworthCIM99}
Donald Aingworth, Chandra Chekuri, Piotr Indyk, and Rajeev Motwani.
\newblock Fast estimation of diameter and shortest paths (without matrix
  multiplication).
\newblock {\em SIAM Journal on Computing}, 28(4):1167--1181, 1999.

\bibitem{BackursI15}
Arturs Backurs and Piotr Indyk.
\newblock Edit distance cannot be computed in strongly subquadratic time
  (unless {SETH} is false).
\newblock In {\em Proc. 47th ACM Symposium on Theory of Computing (STOC)},
  pages 51--58, 2015.

\bibitem{BosekLSZ14}
Bartlomiej Bosek, Dariusz Leniowski, Piotr Sankowski, and Anna Zych.
\newblock Online bipartite matching in offline time.
\newblock In {\em Proc. 55th IEEE Symposium on Foundations of Computer Science
  (FOCS)}, pages 384--393, 2014.

\bibitem{BringmannK15}
Karl Bringmann and Marvin K{\"{u}}nnemann.
\newblock Quadratic conditional lower bounds for string problems and dynamic
  time warping.
\newblock In {\em Proc. 56th IEEE Symposium on Foundations of Computer Science
  (FOCS)}, pages 79--97, 2015.

\bibitem{CairoGR16}
Massimo Cairo, Roberto Grossi, and Romeo Rizzi.
\newblock New bounds for approximating extremal distances in undirected graphs.
\newblock In {\em Proc. 27th ACM/SIAM Symposium on Discrete Algorithms (SODA)},
  pages 363--376, 2016.

\bibitem{ChechikLRSTW14}
Shiri Chechik, Daniel~H. Larkin, Liam Roditty, Grant Schoenebeck, Robert~Endre
  Tarjan, and Virginia~Vassilevska Williams.
\newblock Better approximation algorithms for the graph diameter.
\newblock In {\em Proc. 25th ACM/SIAM Symposium on Discrete Algorithms (SODA)},
  pages 1041--1052, 2014.

\bibitem{Cormen09book}
Thomas~H. Cormen, Charles~E. Leiserson, Ronald~L. Rivest, and Clifford Stein.
\newblock {\em Introduction to Algorithms, Third Edition}.
\newblock The MIT Press, 3rd edition, 2009.

\bibitem{CyganGS15}
Marek Cygan, Harold~N. Gabow, and Piotr Sankowski.
\newblock Algorithmic applications of baur-strassen's theorem: Shortest cycles,
  diameter, and matchings.
\newblock {\em Journal of the ACM}, 62(4):28, 2015.

\bibitem{HenzingerKNS15}
Monika Henzinger, Sebastian Krinninger, Danupon Nanongkai, and Thatchaphol
  Saranurak.
\newblock Unifying and strengthening hardness for dynamic problems via the
  online matrix-vector multiplication conjecture.
\newblock In {\em Proc. 47th ACM Symposium on Theory of Computing (STOC)},
  pages 21--30, 2015.

\bibitem{ImpagliazzoP01}
Russell Impagliazzo and Ramamohan Paturi.
\newblock On the complexity of k-sat.
\newblock {\em Journal of Computer and System Sciences}, 62(2):367--375, 2001.

\bibitem{ImpagliazzoPZ01}
Russell Impagliazzo, Ramamohan Paturi, and Francis Zane.
\newblock Which problems have strongly exponential complexity?
\newblock {\em Journal of Computer and System Sciences}, 63(4):512--530, 2001.

\bibitem{KelnerLOS14}
Jonathan~A. Kelner, Yin~Tat Lee, Lorenzo Orecchia, and Aaron Sidford.
\newblock An almost-linear-time algorithm for approximate max flow in
  undirected graphs, and its multicommodity generalizations.
\newblock In {\em Proc. 25th ACM/SIAM Symposium on Discrete Algorithms (SODA)},
  pages 217--226, 2014.

\bibitem{KopelowitzPP16}
Tsvi Kopelowitz, Seth Pettie, and Ely Porat.
\newblock Higher lower bounds from the 3sum conjecture.
\newblock In {\em Proc. 27th ACM/SIAM Symposium on Discrete Algorithms (SODA)},
  pages 1272--1287, 2016.

\bibitem{LeeS14}
Yin~Tat Lee and Aaron Sidford.
\newblock Path finding methods for linear programming: Solving linear programs
  in {\~{o}}(vrank) iterations and faster algorithms for maximum flow.
\newblock In {\em Proc. 55th IEEE Symposium on Foundations of Computer Science
  (FOCS)}, pages 424--433, 2014.

\bibitem{MadryThesis}
Aleksander Madry.
\newblock {\em From Graphs to Matrices, and Back: New Techniques for Graph
  Algorithms}.
\newblock PhD thesis, Massachusetts Institute of Technology, 6 2011.

\bibitem{Madry13}
Aleksander Madry.
\newblock Navigating central path with electrical flows: From flows to
  matchings, and back.
\newblock In {\em Proc. 54th IEEE Symposium on Foundations of Computer Science
  (FOCS)}, pages 253--262, 2013.

\bibitem{NeimanS13}
Ofer Neiman and Shay Solomon.
\newblock Simple deterministic algorithms for fully dynamic maximal matching.
\newblock In {\em Proc. 45th ACM Symposium on Theory of Computing (STOC)},
  pages 745--754, 2013.

\bibitem{Patrascu10}
Mihai Patrascu.
\newblock Towards polynomial lower bounds for dynamic problems.
\newblock In {\em Proc. 42nd ACM Symposium on Theory of Computing (STOC)},
  pages 603--610, 2010.

\bibitem{PatrascuW10}
Mihai P{\v a}tra{\c s}cu and Ryan Williams.
\newblock On the possibility of faster {SAT} algorithms.
\newblock In {\em Proc. 21st ACM/SIAM Symposium on Discrete Algorithms (SODA)},
  pages 1065--1075, 2010.

\bibitem{RodittyW13}
Liam Roditty and Virginia~Vassilevska Williams.
\newblock Fast approximation algorithms for the diameter and radius of sparse
  graphs.
\newblock In {\em Proc. 45th ACM Symposium on Theory of Computing (STOC)},
  pages 515--524, 2013.

\bibitem{RodittyZ11}
Liam Roditty and Uri Zwick.
\newblock On dynamic shortest paths problems.
\newblock {\em Algorithmica}, 61(2):389--401, 2011.
\newblock See also ESA'04.

\bibitem{Sankowski07}
Piotr Sankowski.
\newblock Faster dynamic matchings and vertex connectivity.
\newblock In {\em Proc. 18th ACM/SIAM Symposium on Discrete Algorithms (SODA)},
  pages 118--126, 2007.

\bibitem{Sherman13}
Jonah Sherman.
\newblock Nearly maximum flows in nearly linear time.
\newblock In {\em Proc. 54th IEEE Symposium on Foundations of Computer Science
  (FOCS)}, pages 263--269, 2013.

\bibitem{VassilevskaW09}
Virginia Vassilevska and Ryan Williams.
\newblock Finding, minimizing, and counting weighted subgraphs.
\newblock In {\em Proc. 41st ACM Symposium on Theory of Computing (STOC)},
  pages 455--464, 2009.

\end{thebibliography}

\end{document}